\documentclass[conference, letterpaper, romanappendices]{IEEEtran}
\usepackage{graphics}
\usepackage{cite}
\usepackage[pdftex]{graphicx}
\usepackage{epstopdf}
\usepackage{epsfig}
\usepackage{latexsym}
\usepackage{amsfonts}
\usepackage{calc}
\usepackage{url}
\usepackage{enumerate}
\usepackage{color}
\usepackage[tbtags]{amsmath}
\usepackage{amssymb}
\usepackage{upref}
\usepackage{dsfont}
\usepackage{multirow}
\usepackage{booktabs}
\usepackage{bigstrut}
\usepackage{rotating}
\usepackage{bbm}
\usepackage{breqn}
\usepackage{nth}
\usepackage{bm}

\newtheorem{theorem}{Theorem}

\newtheorem{lemma}{Lemma}

\newtheorem{proof}{Proof}

\begin{document}

\title{Successive Local and Successive Global Omniscience}

%\author{\IEEEauthorblockN{Anoosheh Heidarzadeh and Alex Sprintson}
%\IEEEauthorblockA{Department of Electrical and Computer Engineering\\ Texas A\&M University\\ Email: \{anoosheh,spalex\}@tamu.edu
%}
%%\thanks{This work was supported by the National Science Foundation under Grant No. CNS-0954153 and by the AFOSR under contract No. FA9550-13-1-0008.}
%}
%
%% 
%%\authorblockN{Anoosheh Heidarzadeh and Alex Sprintson}
%%\authorblockA{Department of Electrical and Computer Engineering, Texas A\&M University, College Station, TX 77843 USA (E-mail: anoosheh@tamu.edu; spalex@tamu.edu)}

\author{Anoosheh Heidarzadeh and Alex Sprintson\\ Texas A\&M University, College Station, TX 77843 USA} 

\maketitle
\thispagestyle{empty}  

\begin{abstract}
This paper considers two generalizations of the cooperative data exchange problem, referred to as the \emph{successive local omniscience} (SLO) and the \emph{successive global omniscience} (SGO). The users are divided into $\ell$ nested sub-groups. Each user initially knows a subset of packets in a ground set $X$ of size $k$, and all users wish to learn all packets in $X$. The users exchange their packets by broadcasting coded or uncoded packets. In SLO or SGO, in the $l$th ($1\leq l\leq \ell$) round of transmissions, the $l$th smallest sub-group of users need to learn all packets they collectively hold or all packets in $X$, respectively. The problem is to find the minimum sum-rate (i.e., the total transmission rate by all users) for each round, subject to minimizing the sum-rate for the previous round. To solve this problem, we use a linear-programming approach. For the cases in which the packets are randomly distributed among users, we construct a system of linear equations whose solution characterizes the minimum sum-rate for each round with high probability as $k$ tends to infinity. Moreover, for the special case of two nested groups, we derive closed-form expressions, which hold with high probability as $k$ tends to infinity, for the minimum sum-rate for each round. 
\end{abstract}

% , with probability approaching $\boldsymbol{1}$ as $\boldsymbol{k}$ tends to infinity

% where different subsets of users have different objectives in different rounds of transmissions

%  approaching $1$ (as the number of packets in $\boldsymbol{X}$ tends to infinity)

\section{Introduction}
The problem of minimizing the total communication rate between users in a distributed system for achieving \emph{omniscience}, i.e., to learn all pieces of knowledge being distributed among users, is a fundamental problem in information theory \cite{CAZDLS:2016}. An example of such problems is \emph{cooperative data exchange} (CDE) \cite{RCS:2007}. The CDE problem considers a group of $n$ users where each user has a subset of packets in a ground set $X$ of size $k$, and wants the rest of packets in $X$. To exchange all their packets, the users broadcast (uncoded or coded) packets over a shared lossless channel, and the problem is to find the minimum sum-rate (i.e., the total transmission rate by all users) such that each user learns all packets in $X$.

%  and bounds on the minimum sum-rate were established in~\cite{RSS:2010}

% DCLKS:16,HS3:2016

Originally, CDE was proposed in \cite{RCS:2007} for a broadcast network, and was later generalized for arbitrary networks in~\cite{CW:2014,GL:2012}. Several solutions for CDE were proposed in~\cite{SSBR:2010, SSBR2:2010, MPRGR:2016}. Several extensions of CDE were also studied, e.g., in~\cite{CW:2011,YSZ:2014,YS:2014,HS:2015,HS1:2016,CH:2016}. Moreover, it was shown in~\cite{CW:2014}, and more recently in~\cite{HS:2015} and~\cite{HS1:2016}, that a solution to CDE and CDE with erasure/error correction can be characterized in closed-form when packets are randomly distributed among users. 

% Scenarios with various transmission costs were studied in~\cite{OS:2011,TSS:2011}. Scenarios providing secrecy and weak security were considered in~\cite{CW:2011,CH:2016}, and~\cite{YS:2013,YSZ:2014}, respectively, and CDE with error/erasure correction was studied in~\cite{YS:2014, HS:2015,HS1:2016}. 

% Game-theoretic perspectives on CDE were studied in~\cite{DCLKS:16,HS3:2016}.

In real-world applications, e.g., when users have different priority, or they are in different locations, different groups of users may have different objectives regarding achieving omniscience in different rounds of transmissions. Addressing such scenarios, in the literature, there have been two extensions of CDE: (i) successive omniscience (SO) \cite{CAZDLS:2016}, and (ii) CDE with priority (CDEP) \cite{HS2:2016}. In SO, in the first round of transmissions, a given subset of users achieve \emph{local omniscience}, i.e., they learn all the packets they collectively have, and in the second round of transmissions, all users in the system achieve omniscience. In CDEP, in the first round of transmissions, a given subset of users achieve \emph{global omniscience}, i.e., they learn all the packets in $X$, and in the second round of transmissions, the rest of the users achieve omniscience. 

In this work, we consider extensions of SO and CDEP scenarios, referred to as the \emph{successive local omniscience} (SLO), and \emph{successive global omniscience} (SGO), where the users are divided into $\ell$ ($1\leq \ell\leq n$) nested sub-groups. In the $l$th ($1\leq l\leq \ell$) round of transmissions, the $l$th smallest sub-group of users need to achieve local or global omniscience in SLO or SGO, respectively. The problem is to find the minimum sum-rate for each round, subject to minimizing the sum-rate for the previous round. 

\subsection{Our Contributions}
We use a multi-objective linear programming (MOLP)  with $O(2^{n}\cdot \ell)$ constraints and $n\ell$ variables to solve the SLO and SGO problems for any arbitrary problem instance. For any instance where the packets are randomly distributed among users, we identify a system of $n\ell$ linear equations with $n\ell$ variables whose solution characterizes the minimum sum-rate for each round with high probability as $k$ tends to infinity. Moreover, for the special case of two nested groups, we derive closed-form expressions for the minimum sum-rate for each round which hold with high probability as $k$ tends to infinity.  

\section{Problem Setup}\label{sec:PS}
Consider a group of $n$ users $N=\{1,\dots,n\}$ and a set of $k$ packets $X=\{x_1,x_2,\dots,x_k\}$. Initially, each user $i\in N$ holds a subset $X_i$ of packets in $X$. Assume, without loss of generality, that $X=\cup_{i\in N} X_i$. Suppose that the index set of packets available at each user is known by all other users. Assume that each packet can be partitioned into an arbitrary (but the same for all packets) number of chunks of equal size. The ultimate objective of all users is to achieve \emph{omniscience}, i.e., to learn all the packets in $X$, via broadcasting (coded or uncoded) chunks over an erasure/error-free channel. This scenario is known as the \emph{cooperative data exchange} (CDE). In CDE, the problem is to find the minimum \emph{sum-rate}, i.e., the total transmission rate by all users, where the \emph{transmission rate} of each user is the total number of chunks being transmitted by the user, normalized by the number of chunks per packet.   

% (a.k.a.~successive omniscience in \cite{CAZDLS:2016})
% (a.k.a.~CDE with priority in \cite{})

In this work, we consider two generalizations of CDE where the transmissions are divided into multiple rounds, and different groups of users have different objectives in each round: (i) \emph{successive local omniscience} (SLO), and (ii) \emph{successive global omniscience} (SGO). Fix an arbitrary integer $1\leq \ell\leq n$, and arbitrary integers $1< n_1<n_2<\dots <n_{\ell}=n$. Suppose that the set $N$ of $n$ users is divided into $\ell$ nested sub-groups $\emptyset\neq N_1\subsetneq N_2\subsetneq\dots\subsetneq N_{\ell}=N$, where $N_{l}\triangleq\{1,\dots,n_l\}$, $\forall 1\leq l\leq \ell$. Let $X^{(l)}$ be the set of all packets that all users in $N_l$ collectively hold, i.e., $X^{(l)}\triangleq\cup_{i\in N_l} X_i$, $\forall 1\leq l\leq \ell$. (Note that $X^{(\ell)}=X$.) Let $\overline{X}^{(l)}_i$ be the set of all packets in $X^{(l)}$ that user $i$ does not hold, i.e., $\overline{X}^{(l)}_i\triangleq X^{(l)} \setminus X_i$, $\forall i\in N_l$, $\forall 1\leq l\leq \ell$. (Note that $\overline{X}^{(\ell)}_i=X\setminus X_i$, $\forall i\in N$.) For the ease of notation, let $\overline{X}_i\triangleq \overline{X}^{(\ell)}_i$, $\forall i\in N$.  

In SLO, in the $l$th ($1\leq l\leq \ell$) round of transmissions, the objective of all users in $N_l$ is to achieve \emph{local omniscience}, i.e., to learn all packets in $X^{(l)}$, without any help from the users in $N\setminus N_l$. (Each user $i\in N_l$ needs to learn all packets in $\overline{X}^{(l)}_i$.) In SGO, in the $l$th ($1\leq l\leq \ell$) round of transmissions, the objective of all users in $N_l$ is to achieve \emph{global omniscience}, i.e., to learn all packets in $X$, possibly with the help of users in $N\setminus N_l$. (Each user $i\in N_l$ needs to learn all packets in $\overline{X}_i$.) Note that, for $\ell=1$, SLO and SGO reduce to CDE. 

% ($r^{(l)}_i\geq 0$, $\forall i\in N$, $\forall 1\leq l\leq \ell$)

% , , $\forall 1\leq l\leq\ell$

Let $r^{(l)}_i$ be the transmission rate of user $i$ in the $l$th round. Let $r^{(l)}_S\triangleq\sum_{i\in S} r^{(l)}_i$, $\forall S\subseteq N$, be the sum-rate of all users in $S$ in the $l$th round, and let $r^{(0)}_S\triangleq 0$, $\forall S\subseteq N$. In SLO and SGO, the problem is to find the minimum $r^{(l)}_N$ for each $1\leq l\leq \ell$, subject to minimizing $r^{(l-1)}_N$. (Note that this is equivalent to finding the minimum $r^{(l)}_N$ for each $1\leq l\leq \ell$, subject to minimizing $r^{(m)}_N$ for all $1\leq m<l$.) Our goal is to solve this problem for any given problem instance $\{X_i\}$.

\section{Arbitrary Problem Instances}\label{sec:MR}

% . 

% ($1\leq l\leq\ell$)

% , N_{l-1}\subset S, N_l\not\subset S

Using similar techniques previously used for CDE in \cite{CW:2014}, each round in SLO and SGO can be reduced to a multicast network coding scenario. Thus, the necessary and sufficient conditions for achieving \emph{local} and \emph{global} omniscience in the $l$th round are given by the following cut-set constraints: \[\sum_{m=1}^{l} r^{(m)}_S\geq \left|\cap_{i\in N_l\setminus S} \overline{X}^{(l)}_i \right|, \forall S\subsetneq N_l,\] and \[\sum_{m=1}^{l} r^{(m)}_S\geq \left|\cap_{i\in N\setminus S} \overline{X}_i \right|, \forall S\subsetneq N, N_{l-1}\subset S, N_{l}\not\subset S,\] respectively. (We sketch the proof of necessity of these constraints in the proofs of Theorems~\ref{thm:ArbitrarySLO} and~\ref{thm:ArbitrarySGO}, and omit the proof of their sufficiency which relies on the standard network-coding argument \cite{CW:2014}.) Based on these constraints, for any instance $\{X_i\}$, one can find a solution to SLO or SGO by solving a multi-objective linear programming (MOLP) (see Theorem~\ref{thm:ArbitrarySLO} and Theorem~\ref{thm:ArbitrarySGO}). 

The special case of the following results for $\ell=2$ were previously presented in \cite{CAZDLS:2016} and \cite{HS2:2016}.

\begin{theorem}\label{thm:ArbitrarySLO}
For any instance $\{X_i\}$, any solution to the SLO problem is a solution to the following MOLP (and vice versa):
\begin{eqnarray}\label{eq:LPSLO}
\mathrm{min} && \hspace{-1.25em} r^{(\ell)}_N \\[-0.25em]  \nonumber \label{eq:SLOC1}
 && \hspace{-1.25em} \dots\\[-0.25em] 
\label{eq:SLOO2} \mathrm{min} && \hspace{-1.25em} r^{(2)}_N \\ \label{eq:SLOO1} \mathrm{min} && \hspace{-1.25em} r^{(1)}_N \\[-0.25em]  
\mathrm{s.t.} 
&& \label{eq:SLOC2} \hspace{-1.25em} r^{(1)}_S\geq \bigg|\bigcap_{i\in N_1\setminus S} \overline{X}^{(1)}_i\bigg|, \forall S\subsetneq N_1\\[-0.25em] 
&& \label{eq:SLOC3} \hspace{-1.25em} \sum_{l=1}^{2} r^{(l)}_S\geq \bigg|\bigcap_{i\in N_{2}\setminus S}\overline{X}^{(2)}_i\bigg|, \forall S\subsetneq N_{2}\hspace{1.75em}\\ \nonumber
&& \hspace{-1.25em} \dots \\
&& \label{eq:SLOC4} \hspace{-1.25em} \sum_{l=1}^{\ell} r^{(l)}_S\geq \bigg|\bigcap_{i\in N_{\ell}\setminus S}\overline{X}^{(\ell)}_i\bigg|, \forall S\subsetneq N_{\ell}\hspace{1.5em}\\[-0.25em]  \nonumber
%&& \hspace{-1.25em} (r_i^{(l)}=0, \forall i\in N\setminus M_l, \forall 1\leq l\leq \ell) \\ \nonumber
&&  \hspace{-1.25em} (r^{(l)}_i\geq 0, \forall i\in N_l, \forall 1\leq l\leq \ell)\\ \nonumber
&& \hspace{-1.25em} (r^{(l)}_i= 0, \forall i\in N\setminus N_l, \forall 1\leq l\leq \ell)
\end{eqnarray} 
\end{theorem}

\begin{proof}[Proof (Sketch)]
In the first round, all users in $N_1$ need to learn $X^{(1)}$. Thus, for any (proper) subset $S$ of users in $N_1$, the corresponding constraint $r^{(1)}_S\geq |\cap_{i\in N_1\setminus S} \overline{X}^{(1)}_i|$ is necessary. This is due to the fact that, for any $S\subsetneq N_1$, each user $i\in N_1\setminus S$ needs to learn $\overline{X}^{(1)}_i$. This yields the constraints in~\eqref{eq:SLOC2}. For any other $S$, the corresponding constraint is, however, unnecessary. This comes from the fact that $r^{(1)}_S =r^{(1)}_{S\cap N_1} + r^{(1)}_{S\setminus N_1}= r^{(1)}_{S\cap N_1} \geq |\cap_{i\in N_1\setminus (S\cap N_1)} \overline{X}^{(1)}_i|=|\cap_{i\in N_1\setminus S} \overline{X}^{(1)}_i|$. %(Note that only users in $N_1$ may transmit in the first round.)

% , but not including all users in $N_1$

% (i) for any $S\subset N_1$, $r^{(1)}_S+r^{(2)}_S\geq r^{(1)}_S\geq |\cap_{i\in N_1\setminus S} \overline{X}^{(1)}_i|\geq |\cap_{i\in N_2\setminus S} \overline{X}^{(2)}_i|$, and (ii) 

In the second round, all users in $N_2$ need to learn $X^{(2)}$. Similarly as above, for any (proper) subset $S$ of users in $N_2$, the corresponding constraint $r^{(1)}_S+r^{(2)}_S\geq |\cap_{i\in N_2\setminus S} \overline{X}^{(2)}_i|$ imposes a necessary constraint, and hence the constraints in~\eqref{eq:SLOC3}. However, for any other $S$, the corresponding constraint is unnecessary. This is because $r^{(1)}_S+r^{(2)}_S = r^{(1)}_{S\cap N_2}+ r^{(2)}_{S\cap N_2} \geq |\cap_{i\in N_2\setminus (S\cap N_2)} \overline{X}^{(2)}_i|=|\cap_{i\in N_2\setminus S} \overline{X}^{(2)}_i|$. %(Note that only users in $N_2$ may transmit in the second round.) 

Repeating the same argument as above, it follows that the necessary constraints for all users in $N_l$ to learn $X^{(l)}$ in the $l$th round are $\sum_{1\leq m\leq l} r^{(m)}_S\geq |\cap_{i\in N_l\setminus S} \overline{X}^{(l)}_i|$, $\forall S\subsetneq N_l$.
\end{proof}

\begin{theorem}\label{thm:ArbitrarySGO}
For any instance $\{X_i\}$, any solution to the SGO problem is a solution to the following MOLP (and vice versa):
\begin{eqnarray}\label{eq:LPSGO}
\mathrm{min} && \hspace{-1.25em} r^{(\ell)}_N \\[-0.25em]  \nonumber \label{eq:SGOC1}
 && \hspace{-1.25em} \dots\\[-0.25em] 
\mathrm{min} && \hspace{-1.25em} r^{(2)}_N \\ \mathrm{min} && \hspace{-1.25em} r^{(1)}_N \\[-0.25em]  
\mathrm{s.t.} 
&& \label{eq:SGOC2} \hspace{-1.25em} r^{(1)}_S\geq \bigg|\bigcap_{i\in N\setminus S} \overline{X}_i\bigg|, \forall S\subsetneq N, N_1\not\subset S\\[-0.25em]  
&& \label{eq:SGOC3} \hspace{-1.25em} \sum_{l=1}^{2} r^{(l)}_S\geq \bigg|\bigcap_{i\in N\setminus S}\overline{X}_i\bigg|, \forall S\subsetneq N, N_1\subset S, N_2\not\subset S \hspace{1.75em}\\ 
\nonumber
&& \hspace{-1.25em} \dots \\
&& \label{eq:SGOC4} \hspace{-1.25em} \sum_{l=1}^{\ell} r^{(l)}_S\geq \bigg|\bigcap_{i\in N\setminus S}\overline{X}_i\bigg|, \forall S\subsetneq N, N_{\ell-1}\subset S, N_{\ell}\not\subset S \hspace{2em}\\[-0.25em]  \nonumber
&&  \hspace{-1.25em} (r^{(l)}_i\geq 0, \forall i\in N, \forall 1\leq l\leq \ell)
\end{eqnarray} 
\end{theorem}

\begin{proof}[Proof (Sketch)]
In the first round, all users in $N_1$ need to learn $X$ and none of users in $N\setminus N_1$ need to learn $X$. Thus, for any (proper) subset $S$ of users in $N$ not containing $N_1$, the corresponding constraint $r^{(1)}_S\geq |\cap_{i\in N\setminus S} \overline{X}_i|$ is necessary, and hence~\eqref{eq:SGOC3}. For any $S$ containing $N_1$, the corresponding constraint is however unnecessary since $N\setminus S$ consists only of users which need not learn $X$ in the first round.   

In the second round, all users in $N_2$ need to learn $X$. Since all users in $N_1$ learn $X$ in the first round, for any $S$ containing $N_1$ but not $N_2$, the corresponding constraint $r^{(1)}_S+r^{(2)}_S\geq |\cap_{i\in N\setminus S} \overline{X}_i|$ imposes a necessary constraint, and hence~\eqref{eq:SGOC4}. For any $S$ containing $N_2$, the corresponding constraint is unnecessary since $N\setminus S$ consists only of users not in $N_2$, and none of such users need to learn $X$ in the second round. Note that, for any (proper) $S$ not containing $N_1$, the corresponding constraint is redundant since $N\setminus S$ includes some user(s) in $N_1$, and such users learn $X$ in the first round. 

By using a similar argument as above, it follows that the necessary constraints for all users in $N_l$ to learn $X$ in the $l$th round are $\sum_{1\leq m\leq l} r^{(m)}_S\geq |\cap_{i\in N\setminus S} \overline{X}_i|$, $\forall S\subsetneq N$, $N_{l-1}\subset S$, $N_l\not\subset S$.
\end{proof}

\section{Random Packet Distribution}
In this section, we assume that each packet is available at each user, independently from other packets and other users, with probability $0<p<1$. (This model is referred to as the \emph{random packet distribution} in~\cite{CW:2014,HS:2015,HS1:2016}.) 

Theorems~\ref{thm:SLORandom} and ~\ref{thm:SGORandom} characterize, with probability approaching $1$ (w.p.~$\rightarrow 1$) as $k$ tends to infinity ($k\rightarrow \infty$), a solution to SLO and SGO by a system of linear equations (SLE) for any \emph{random} problem instance under the assumption above. 

% (referred to as the \emph{random packet distribution} \cite{CW:2014})

% (a.k.a.~the \emph{random packet distribution assumption} in~\cite{CW:2014}) 

\begin{theorem}\label{thm:SLORandom}
For any random instance $\{X_i\}$, w.p.~$\rightarrow 1$ as $k\rightarrow\infty$, a solution to the SLO problem is given by the following SLE:\vspace{-0.5em}  
\begin{eqnarray}
&&  \label{eq:SLE1E1}\hspace{-2em} r^{(1)}_{N_1\setminus \{i\}} = \left|\overline{X}^{(1)}_{i}\right|, \forall i\in N_1\\
&& \label{eq:SLE1E2}\hspace{-2em} r^{(l)}_{N_{l,j}} = \bigg|\bigcap_{i\in N_l\setminus N_{l,j}}\overline{X}^{(l)}_{i}\bigg|, \forall 1\leq j\leq d_l, \forall 1<l\leq \ell\\ 
&&  \label{eq:SLE1E3}\hspace{-2em} r^{(l)}_i = 0, \forall i\in N\setminus N_l, \forall 1\leq l\leq \ell\\
&&  \label{eq:SLE1E4}\hspace{-2em} r^{(l)}_i = 0, \forall i\in N_{l-1}, \forall 1<l\leq\ell 
\end{eqnarray} where $d_l\triangleq n_l-n_{l-1}$ for all $1<l\leq \ell$ and $N_{l,j}\triangleq \{n_{l-1}+1,\dots,n_{l-1}+j\}$ for all $1<l\leq \ell$ and all $1\leq j\leq d_l$.
\end{theorem}

\begin{theorem}\label{thm:SGORandom}
For any random instance $\{X_i\}$, w.p.~$\rightarrow 1$ as $k\rightarrow\infty$, a solution to the SGO problem is given by the following SLE:\vspace{-0.5em} 
\begin{eqnarray}
&&  \label{eq:SLE2E1}\hspace{-2.5em} \sum_{m=1}^{m_l} r^{(m)}_{N_l\setminus \{j_l\}} = \bigg|\bigcap_{i\not\in N_l\setminus \{j_l\}}\overline{X}_{i}\bigg|, \forall 1\leq l<\ell\\
&& \label{eq:SLE2E2}\hspace{-2.5em} \sum_{m=1}^{l} r^{(m)}_{N\setminus \{i\}} = \left|\overline{X}_i\right|, \forall i\in N_l\setminus N_{l-1}, \forall 1\leq l\leq \ell\\
&&  \label{eq:SLE2E3}\hspace{-2.5em} r^{(l)}_i = 0, \forall i\in N\setminus N_{l-1}, \forall 1< l\leq \ell\\
&&  \label{eq:SLE2E4}\hspace{-2.5em} r^{(l)}_i = r^{(l)}_j, \forall i,j\in N_{l-1}, \forall 1<l\leq\ell  
\end{eqnarray} where for any $1\leq l< \ell$, $j_l\in N_{m_{l}}\setminus N_{m_{l}-1}$ (for some $1\leq m_{l}\leq l$) such that $\sum_{i\in N_l\setminus\{j_l\}}\left|\overline{X}_i\right|+\left|\cap_{i\not\in N_l\setminus \{j_l\}} \overline{X}_i\right|\geq \sum_{i\in N_l\setminus\{j\}} \left|\overline{X}_i\right|+\left|\cap_{i\not\in N_l\setminus\{j\}} \overline{X}_i\right|$ for all $j\in N_l$.  
\end{theorem}

Note that~\eqref{eq:SLE1E3} and~\eqref{eq:SLE1E4} imply that for SLO, in the first round, only users in $N_1$ may transmit, and in each round $l>1$, only users in $N_l\setminus N_{l-1}$ need to transmit. A closer look at~\eqref{eq:SLE1E1} and~\eqref{eq:SLE1E2} reveals that for SLO, only the users in $N_1$ may need to transmit at fractional rates, and it suffices for the rest of users in $N\setminus N_1$ to transmit at integral rates. Moreover,~\eqref{eq:SLE2E3} and~\eqref{eq:SLE2E4} imply that for SGO, in the first round, all users in $N$ may transmit, but in each round $l>1$, only users in $N_{l-1}$ need to transmit, and they all can transmit at the same rate.   

Theorems~\ref{thm:SLORandomSpecial} and~\ref{thm:SGORandomSpecial} give a closed-form solution to the SLE's in~\eqref{eq:SLE1E1}-\eqref{eq:SLE1E4} and~\eqref{eq:SLE2E1}-\eqref{eq:SLE2E4} for the special case of $\ell=2$. 

\begin{theorem}\label{thm:SLORandomSpecial}
For any random instance $\{X_i\}$, w.p.~$\rightarrow 1$ as $k\rightarrow\infty$, a solution to the SLO problem for $\ell=2$ is given by 
\begin{equation*}\label{eq:SLOr1}
\tilde{r}^{(1)}_i=\left\{
\begin{array}{ll}
\hspace{-0.25em}\frac{1}{n_1-1}\sum_{j\in N_1} |\overline{X}^{(1)}_j|-|\overline{X}^{(1)}_i|, & \hspace{-0.25em} i\in N_1\\ 
\hspace{-0.25em}0, & \hspace{-0.25em} i\not\in N_1
\end{array} 
\right.\end{equation*} and
\begin{equation*}\label{eq:SLOr2}
\tilde{r}^{(2)}_i=\left\{
\begin{array}{ll}
\hspace{-0.25em}0, & \hspace{-0.25em} i\in N_1\\ 
\hspace{-0.25em}|\cap_{j\not\in N_{2,i-n_1}}\overline{X}^{(2)}_j|-|\cap_{j\not\in N_{2,i-n_1-1}}\overline{X}^{(2)}_j|, & \hspace{-0.25em} i\not\in N_1
\end{array} 
\right.\end{equation*} 
\end{theorem}

\begin{theorem}\cite[Theorem~2]{HS2:2016}\label{thm:SGORandomSpecial}
For any random instance $\{X_i\}$, w.p.~$\rightarrow 1$ as $k\rightarrow\infty$, a solution to the SGO problem for $\ell=2$ is given by 
\begin{equation*}\label{eq:SGOr1}
\hspace{0.5em}\tilde{r}^{(1)}_i=\left\{
\begin{array}{ll}
\hspace{-0.25em}\frac{1}{n_1-1}\hspace{-0.25em}\left(\sum_{j\in M} |\overline{X}_j|+|\cap_{j\not\in M}\overline{X}_j|\right)-|\overline{X}_i|, & \hspace{-0.25em} i\in N_1\\ 
\hspace{-0.25em}\frac{1}{n-n_1}\hspace{-0.25em}\left(\sum_{j\not\in M} |\overline{X}_j|-|\cap_{j\not\in M}\overline{X}_j|\right)-|\overline{X}_i|, & \hspace{-0.25em} i\not\in N_1 
\end{array} 
\right.\end{equation*} and 
\begin{equation*}\label{eq:SGOr2}
\hspace{0.25em}\tilde{r}^{(2)}_i=\left\{
\begin{array}{ll}
\hspace{-0.25em} \frac{\sum_{j\not\in M} \left|\overline{X}_j\right|}{n_1(n-n_1)}\hspace{-0.125em}-\frac{\sum_{j\in M} \left|\overline{X}_j\right|}{n_1(n_1-1)}\hspace{-0.125em}-\frac{\hspace{-0.125em}(n\hspace{-0.075em}-\hspace{-0.075em}1)\left|\cap_{j\not\in M}\overline{X}_j\right|}{n_1(n_1-1)(n-n_1)}, & \hspace{-0.25em} i\in N_1\\ 
 \hspace{-0.25em} 0, & \hspace{-0.25em} i\not\in N_1 
\end{array} 
\right.\end{equation*} where $M\triangleq N_1\setminus\{j_1\}$. 
\end{theorem}

% by using standard concentration inequalities (see, e.g., Lemma~\ref{lem:concentration} in Section~\ref{sec:Proofs}), 

Such closed-form results lead to several interesting observations as follows. First, the minimum required number of chunks per packet for SLO is $n_1-1$, and this quantity for SGO is $\mathrm{LCM}(n_1-1,n-n_1)$. Note that this quantity for CDE is $n-1$ \cite{CW:2014}. Second, for any random instance $\{X_i\}$, the total sum-rate (normalized by the total number of packets ($k$)) is tightly concentrated around\vspace{-0.125em} \[r_{\text{SLO}} \triangleq \frac{n_1(q-q^{n_1})}{(n_1-1)(1-q^{n_1})}+\frac{q^{n_1}-q^n}{1-q^{n}},\vspace{-0.125em}\] and\vspace{-0.125em} \[r_{\text{SGO}} \triangleq \frac{(n-n_1+1)q-(n-n_1)q^{n}-q^{n-n_1+1}}{(n-n_1)(1-q^{n})},\vspace{-0.125em}\] in SLO and SGO, respectively, where $q\triangleq 1-p$. Note that this quantity for CDE is $r_{\text{CDE}}\triangleq\frac{n(q-q^n)}{(n-1)(1-q^n)}$ \cite{CW:2014}. 

Let $e_{\text{SLO}} \triangleq ({r_{\text{SLO}}-r_{\text{CDE}}})/{r_{\text{CDE}}}$ and $e_{\text{SGO}}\triangleq ({r_{\text{SGO}}-r_{\text{CDE}}})/{r_{\text{CDE}}}$ be the excess rate of SLO over CDE and the excess rate of SGO over CDE, respectively. Fig.~\ref{fig:eSLOvseSGO} depicts $e_{\text{SLO}}$ and $e_{\text{SGO}}$ versus $p$ for $n=6$ and $n_1=2,\dots,5$. 

% as functions of $n$, $n_1$, and $p$ 

Comparing $e_{\text{SLO}}(n,n_1,p)$ and $e_{\text{SGO}}(n,n_1,p)$ yields the following non-trivial observations. First, for any $1<n_1\leq \frac{n}{2}$ and any $0<p<1$, $e_{\text{SLO}}\geq e_{\text{SGO}}$, and for any $\frac{n}{2}<n_1\leq n-1$, there exists some $0<p^{*}<1$ such that for any $0<p\leq p^{*}$, $e_{\text{SLO}}\geq e_{\text{SGO}}$ and for any $p^{*}<p<1$, $e_{\text{SLO}}< e_{\text{SGO}}$. Second, for any $1< n_1\leq \frac{n}{2}$, there exists some $0<p_{*}<1$ such that $e_{\text{SLO}}$ decreases as $p$ increases from $0$ to $p_{*}$, and then the trend changes, i.e., $e_{\text{SLO}}$ increases as $p$ increases from $p_{*}$ to $1$. For any $\frac{n}{2}<n_1\leq n-1$, $e_{\text{SLO}}$ decreases as $p$ increases. This is in contrast to $e_{\text{SGO}}$ which, for any $1<n_1\leq n-1$, increases monotonically as $p$ increases. Third, for any $1<n_1\leq n-1$, $e_{\text{SLO}}(n,n_1,p)$ and $e_{\text{SGO}}(n,n-n_1,p)$ converge to the same limit as $p$ approaches $1$.

%This result suggests a duality between SLO and SGO in the limit.

\begin{figure}[t]
\centering
\includegraphics[width=0.45\textwidth]{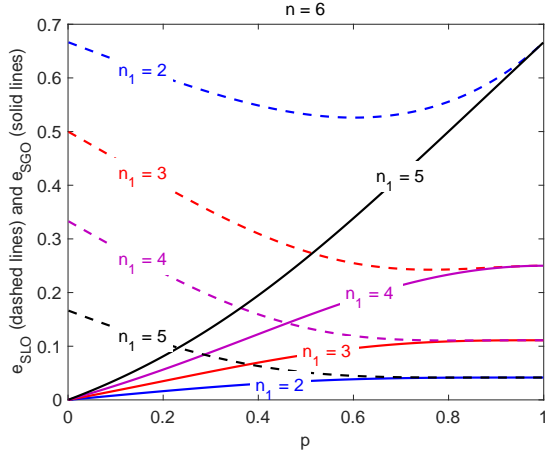}\vspace{-0.5em}
\caption{The excess rate of SLO and SGO over CDE.}\label{fig:eSLOvseSGO}\vspace{-1em}
\end{figure}

\section{Proofs}\label{sec:Proofs}
In this section, we prove Theorem~\ref{thm:SLORandomSpecial}, and refer the reader to \cite[Theorem~2]{HS2:2016} for the proof of Theorem~\ref{thm:SGORandomSpecial}. The proofs of Theorems~\ref{thm:SLORandom} and~\ref{thm:SGORandom} use similar techniques, and are deferred to an extended version of this work due to the space limit.

The proof of Theorem~\ref{thm:SLORandomSpecial} consists of two parts: feasibility of $\{\tilde{r}^{(1)}_i, \tilde{r}^{(2)}_i\}$ with respect to (w.r.t.)~\eqref{eq:SLOC2} and~\eqref{eq:SLOC3} (Lemma~\ref{lem:feasibility}), and optimality of $\{\tilde{r}^{(1)}_i,\tilde{r}^{(2)}_i\}$ w.r.t.~LP~\eqref{eq:SLOO1} and LP~\eqref{eq:SLOO2} (Lemma~\ref{lem:optimality}). 

The proofs rely on the following two lemmas. (The proofs of these lemmas can be found in \cite{HS2:2016}.) 

% of Lemmas~\ref{lem:feasibility} and~\ref{lem:optimality} 
%The following concentration result, \cite[Lemma~5]{}, follows from the law of large numbers under random packet distribution assumption. (The proof is given in \cite{}.) 

\begin{lemma}\cite[Lemma~1]{HS2:2016}\label{lem:concentration}
For any $1\leq l\leq \ell$ and any $S\subsetneq N_l$,  
\[\Bigg|\frac{1}{k}\bigg|\bigcap_{i\in N_l\setminus S} \overline{X}^{(l)}_i\bigg|-z_{|N_l|,|S|}\Bigg|<\epsilon,\] for any $\epsilon>0$, w.p.~$\rightarrow 1$ as $k\rightarrow\infty$, where 
\begin{equation}\label{eq:zms} 
z_{m,s} \triangleq \frac{(1-p)^{m-s}-(1-p)^{m}}{1-(1-p)^{m}},\end{equation} for any $0\leq s<m$.
\end{lemma} 

%The following lemma, \cite[Lemma~2]{}, is crucial for the proof of our results. (The proof is given in \cite{}.) %(see the proof of \cite[Lemma~7]{HS:2015}). 

\begin{lemma} \cite[Lemma~2]{HS2:2016}\label{lem:PV}
For any $0<p<1$ and any $0< s_1<s_2<m$, $\frac{z_{m,s_1}}{s_1}<\frac{z_{m,s_2}}{s_2}$. 
% we have $z_{m,s_1}\cdot s_2<z_{m,s_2}\cdot s_1$.	
\end{lemma}

The rest of the results hold ``w.p.~$\rightarrow 1$ as $k\rightarrow\infty$,'' and hereafter we omit this statement for brevity.

\begin{lemma}\label{lem:feasibility}
$\{\tilde{r}^{(1)}_i,\tilde{r}^{(2)}_i\}$ is feasible w.r.t.~\eqref{eq:SLOC2} and~\eqref{eq:SLOC3}.
\end{lemma}

\begin{proof}
We need to show that: (i) $\tilde{r}^{(1)}_S\geq |\cap_{i\in N_1\setminus S} \overline{X}^{(1)}_i|$, $\forall S\subsetneq N_1$, and (ii) $\tilde{r}^{(1)}_S+\tilde{r}^{(2)}_S\geq |\cap_{i\in N\setminus S} \overline{X}^{(2)}_i|$, $\forall S\subsetneq N$. First, consider the inequality (i). Take an arbitrary $S\subsetneq N_1$. Let $s\triangleq |S|$. First, suppose that $s=n_1-1$. Then, $S = N_1\setminus \{i\}$ for some $i\in N_1$. Since $\tilde{r}^{(1)}_S=|\overline{X}^{(1)}_i|$ and $|\cap_{i\in N_1\setminus S} \overline{X}^{(1)}_i| = |\overline{X}^{(1)}_i|$, the inequality (i) holds. Next, suppose that $1\leq s<n_1-1$. Note that $\tilde{r}^{(1)}_S=\frac{s}{n_1-1}\sum_{i\in N_1}|\overline{X}^{(1)}_i|-\sum_{i\in S}|\overline{X}^{(1)}_i|$. By applying Lemma~\ref{lem:concentration}, $\frac{\tilde{r}^{(1)}_S}{k}>(\frac{s}{n_1-1}) z_{n_1,n_1-1}-\epsilon$ and $\frac{1}{k}|\cap_{i\in N_1\setminus S} \overline{X}^{(1)}_i|<z_{n_1,s}+\epsilon$, for any $\epsilon>0$. Thus, the inequality (i) holds so long as $\frac{z_{n_1,n_1-1}}{n_1-1}>\frac{z_{n_1,s}}{s}$, and this inequality follows from Lemma~\ref{lem:PV} since $1\leq s<n_1-1$ (by assumption). 

% (This comes from two facts: $\frac{1}{k}\sum_{i\in N_1} |\overline{X}^{(1)}_i|>n_1\cdot z_{n_1,n_1-1}-\epsilon$, for any $\epsilon>0$, and $\frac{1}{k}\sum_{i\in S} |\overline{X}^{(1)}_i|<s\cdot z_{n_1,n_1-1}+\epsilon$, for any $\epsilon>0$.)

% , $S\not\subset N_1$

Next, consider the inequality (ii). Take an arbitrary $S\subsetneq N$. Let $S_1\triangleq S\cap N_1$ and $S_2\triangleq S\setminus N_1$. Obviously, $\tilde{r}^{(1)}_S+\tilde{r}^{(2)}_S=\tilde{r}^{(1)}_{S_1}+\tilde{r}^{(2)}_{S_2}$. Let $s \triangleq |S_1|$ and $t\triangleq |S_2|$. (Note that $0\leq s\leq n_1$, $0\leq t\leq n-n_1$, and $0<r+s<n$.) First, suppose that $s=0$. There are two sub-cases: (a) $S_2=\{n_1+1,\dots,n_1+t\}$ ($=N_{2,t}$), $1\leq t\leq n-n_1$, and (b) $S_2=\{n_1+i_1,\dots,n_1+i_{t}\}$, $1\leq t< n-n_1$, $1\leq i_1<\dots<i_t\leq n-n_1$, for arbitrary $i_j\in N\setminus N_1$, $1\leq j\leq t$, such that $i_j>j$ for some $1\leq j\leq t$. In the case (a), $\tilde{r}^{(1)}_{S_1}+\tilde{r}^{(2)}_{S_2}=\tilde{r}^{(2)}_{S_2}=\tilde{r}^{(2)}_{N_{2,t}}$ and $|\cap_{j\in N\setminus S}\overline{X}^{(2)}_j| = |\cap_{j\in N\setminus N_{2,t}}\overline{X}^{(2)}_j|=\tilde{r}^{(2)}_{N_{2,t}}$. Thus, the inequality (ii) holds. In the case (b), $\tilde{r}^{(1)}_{S_1}+\tilde{r}^{(2)}_{S_2}=\tilde{r}^{(2)}_{S_2}=\sum_{i\in S_2} |\cap_{j\not\in N_{2,i-n_1}}\overline{X}^{(2)}_j|-\sum_{i\in S_2} |\cap_{j\not\in N_{2,i-n_1-1}}\overline{X}^{(2)}_j|$. Again by applying Lemma~\ref{lem:concentration}, $\frac{\tilde{r}^{(2)}_{S_2}}{k}>\sum_{i\in S_2} z_{n,i}-\sum_{i\in S_2} z_{n,i-1}-\epsilon$ and $\frac{1}{k} |\cap_{j\in N\setminus N_{2,t}}\overline{X}^{(2)}_j|<z_{n,t}+\epsilon$, for any $\epsilon>0$. Thus, the inequality (ii) holds so long as 
\begin{equation}\label{eq:zzz}
\sum_{i\in S_2} z_{n,i}-\sum_{i\in S_2} z_{n,i-1}>z_{n,t}.\end{equation}
By rewriting~\eqref{eq:zzz} according to~\eqref{eq:zms}, it follows that~\eqref{eq:zzz} holds so long as $\sum_{i\in S_2} (1-p)^{t-i}>\sum_{1\leq j\leq t} (1-p)^{t-j}$. The latter inequality holds since $(1-p)^{t-i_j}\geq (1-p)^{t-j}$ for all $1\leq j\leq t$, and $(1-p)^{t-i_j}>(1-p)^{t-j}$ for some $1\leq j\leq t$, noting that $i_j>j$ for some $1\leq j\leq t$ (by assumption). 

% (This is due to two facts: $\frac{1}{k}|\cap_{j\not\in N_{2,i-n_1}}\overline{X}^{(2)}_j|>z_{n,i}-\epsilon$, for any $\epsilon>0$, and $\frac{1}{k}|\cap_{j\not\in N_{2,i-n_1-1}}\overline{X}^{(2)}_j|<z_{n,i-1}+\epsilon$, for any $\epsilon>0$.)

% .  Also, by Lemma~\ref{lem:concentration}

% , and the RHS of (ii) is $|\cap_{j\in N\setminus S}\overline{X}^{(2)}_j|$

Next, suppose that $s=n_1-1$. Note that $S_1 = N_1\setminus\{i\}$ for some $i\in N_1$. There are two sub-cases: (a) $S_2=\{n_1+1,\dots,n\}$, and (b) $S_2=\{n_1+i_1,\dots,n_1+i_t\}$, $0\leq t<n-n_1$, $1\leq i_1<\dots<i_t\leq n-n_1$. (For $t=0$, $S_2=\emptyset$.) In the case (a), $\tilde{r}^{(1)}_{S_1}+\tilde{r}^{(2)}_{S_2}=|\overline{X}^{(1)}_i|+|\cap_{j\not\in N_{2,n-n_1}}\overline{X}^{(2)}_j| = |\overline{X}^{(1)}_i|+|\cap_{j\in N_{1}}\overline{X}^{(2)}_j|$, and $|\cap_{j\in N\setminus S}\overline{X}^{(2)}_j| = |\overline{X}^{(2)}_i|$. Since $|\overline{X}^{(1)}_i|+|\cap_{j\in N_{1}}\overline{X}^{(2)}_j| = |\overline{X}^{(2)}_i|$ for all $i\in N_1$, the inequality (ii) holds. In the case (b), $\tilde{r}^{(1)}_{S_1}+\tilde{r}^{(2)}_{S_2} = |\overline{X}^{(1)}_i|+\sum_{i\in S_2} |\cap_{j\not\in N_{2,i-n_1}}\overline{X}^{(2)}_j|-\sum_{i\in S_2} |\cap_{j\not\in N_{2,i-n_1-1}}\overline{X}^{(2)}_j|$. Again by applying Lemma~\ref{lem:concentration}, $\frac{1}{k}(\tilde{r}^{(1)}_{S_1}+\tilde{r}^{(2)}_{S_2})>z_{n_1,n_1-1} + \sum_{i\in S_2} z_{n,i} - \sum_{i\in S_2} z_{n,i-1}-\epsilon$ and $\frac{1}{k}|\cap_{j\in N\setminus S}\overline{X}^{(2)}_j|<z_{n,n_1-1+t}+\epsilon$, for any $\epsilon>0$. Thus, the inequality (ii) holds so long as 
\begin{equation}\label{eq:zzzz}
z_{n_1,n_1-1}+\sum_{i\in S_2} z_{n,i}-\sum_{i\in S_2} z_{n,i-1}>z_{n,n_1-1+t}.
\end{equation} Again, rewriting~\eqref{eq:zzzz}, this inequality holds so long as $1+(1-p)^{n-1}+p(1-p)^{n-1}\sum_{i\in S_2} (1-p)^{-i}>(1-p)^{n-n_1-t}+(1-p)^{n_1-1}$. Note that $i_j\geq j$ for all $1\leq j\leq t$. Thus, $\sum_{i\in S_2} (1-p)^{-i}\geq \sum_{1\leq j\leq t} (1-p)^{-j}=\frac{(1-p)^{-t}-1}{p}$. Thus,~\eqref{eq:zzzz} holds so long as $(1-p)^{n}((1-p)^{-n_1-t}-(1-p)^{-t-1})<1-(1-p)^{n_1-1}$. Obviously, $(1-p)^{n}((1-p)^{-n_1-t}-(1-p)^{-t-1})<(1-p)^{n_1+t+1}((1-p)^{-n_1-t}-(1-p)^{-t-1})=(1-p)-(1-p)^{n_1}$ since $n>n_1+t$ (by assumption). Thus,~\eqref{eq:zzzz} holds so long as $(1-p)-(1-p)^{n_1}<1-(1-p)^{n_1-1}$, and this inequality holds since $n_1>1$ (by assumption). 

Now, suppose that $0< s<n_1-1$. (Note that for $t=0$, $S_2=\emptyset$.) Note that $\tilde{r}^{(1)}_{S_1}+\tilde{r}^{(2)}_{S_2}=\frac{s}{n_1-1}\sum_{i\in N_1} |\overline{X}^{(1)}_i|-\sum_{i\in S_1} |\overline{X}^{(1)}_i|$ $+$ $\sum_{i\in S_2} |\cap_{j\not\in N_{2,i-n_1}}\overline{X}^{(2)}_j|$ $-$ $\sum_{i\in S_2} |\cap_{j\not\in N_{2,i-n_1-1}}\overline{X}^{(2)}_j|$. Similarly as above, by applying Lemma~\ref{lem:concentration}, $\frac{1}{k}(\tilde{r}^{(1)}_{S_1}+\tilde{r}^{(2)}_{S_2})>(\frac{s}{n_1-1})z_{n_1,n_1-1}+\sum_{i\in S_2} z_{n,i}-\sum_{i\in S_2} z_{n,i-1}-\epsilon$ and $\frac{1}{k}|\cap_{j\in N\setminus S} \overline{X}^{(2)}_j|<z_{n,s+t}+\epsilon$, for any $\epsilon>0$. Thus, the inequality (ii) holds so long as $(\frac{s}{n_1-1})z_{n_1,n_1-1}+\sum_{i\in S_2} z_{n,i}-\sum_{i\in S_2} z_{n,i-1}>z_{n,s+t}$. Note that $\sum_{i\in S_2} z_{n,i}-\sum_{i\in S_2} z_{n,i-1}>z_{n,n_1-1+t}-z_{n_1,n_1-1}$ (by~\eqref{eq:zzzz}). Thus, the inequality (ii) holds so long as 
\begin{equation}\label{eq:zzz2}
z_{n,n_1-1+t}-\left(\frac{n_1-1-s}{n_1-1}\right)z_{n_1,n_1-1}>z_{n,s+t}.\end{equation} By rewriting~\eqref{eq:zzz2}, this inequality holds so long as $(1-p)^{n-t}((1-p)^{-n_1+1}-(1-p)^{-s})>(\frac{n_1-1-s}{n_1-1})((1-p)-(1-p)^{n_1})$. Obviously, $(1-p)^{-n_1+1}-(1-p)^{-s}>0$ since $s<n_1-1$. Thus, the inequality~\eqref{eq:zzz2} holds so long as $(1-p)^{n_1}((1-p)^{-n_1+1}-(1-p)^{-s})>(\frac{n_1-1-s}{n_1-1})((1-p)-(1-p)^{n_1})$ since $(1-p)^{n-t}\geq (1-p)^{n_1}$. Thus, the inequality (ii) holds so long as $\frac{(1-p)-(1-p)^{n_1-s}}{n_1-1-s}>\frac{(1-p)-(1-p)^{n_1}}{n_1-1}$. Since $n_1>1$ and $0< s<n_1-1$ (by assumption), the latter inequality holds so long as $\frac{1}{m}-\frac{1}{m+1}>\frac{(1-p)^{m}}{m}-\frac{(1-p)^{m+1}}{m+1}$, for any integer $m\geq 1$, and this inequality holds since $(1-p)^{m+1}>1-(m+1)p$ for any integer $m\geq 1$ (by the Bernoulli's inequality).

Lastly, suppose that $s=n_1$. Note that $S_1=N_1$, and $S_2=\{n_1+i_1,\dots,n_1+i_t\}$, $0\leq t<n-n_1$, $1\leq i_1<\dots<i_t\leq n-n_1$. (Note, again, that for $t=0$, $S_2=\emptyset$.) Using similar techniques as above, it can be shown that $\frac{n_1}{n_1-1}z_{n_1,n_1-1}+\sum_{i\in S_2} z_{n,i}-\sum_{i\in S_2} z_{n,i-1}>z_{n,n_1+t}$. By applying this inequality along with an application of Lemma~\ref{lem:concentration}, one can see that $\tilde{r}^{(1)}_{S_1}+\tilde{r}^{(2)}_{S_2} = \frac{1}{n_1-1}\sum_{i\in N_1} |\overline{X}^{(1)}_i|+$ $\sum_{i\in S_2} |\cap_{j\not\in N_{2,i-n_1}}\overline{X}^{(2)}_j|-$ $\sum_{i\in S_2} |\cap_{j\not\in N_{2,i-n_1-1}}\overline{X}^{(2)}_j|\geq |\cap_{j\in N\setminus S}\overline{X}^{(2)}_j|$.  
\end{proof}

% $(1-p)^{-n_1-t}-(1-p)^{-t-1}>0$ since $n_1>1$. Thus, 

% , and the RHS of (ii) is $|\cap_{j\in N\setminus S} \overline{X}^{(2)}_j|$

\begin{lemma}\label{lem:optimality}
$\{\tilde{r}^{(1)}_i,\tilde{r}^{(2)}_i\}$ is optimal w.r.t.~LP~\eqref{eq:SLOO1} and LP~\eqref{eq:SLOO2}. 
\end{lemma}

\begin{proof}
The dual of LP~\eqref{eq:SLOO1} is given by 
\begin{eqnarray}\label{eq:dualLP1}
\mathrm{max.} && \nonumber \hspace{-1.5em} \sum_{S\subsetneq N_1} \bigg|\bigcap_{i\in N_1\setminus S}\hspace{-0.25em}\overline{X}^{(1)}_i\bigg|s_{S}+\hspace{-0.25em}\sum_{S\not\subset N_1} \bigg|\bigcap_{i\in N\setminus S}\overline{X}^{(2)}_i\bigg| s_{S}\hspace{2em}\\ 
\mathrm{s.t.} && \hspace{-1.5em} \label{eq:LP1C1} \sum_{S\subsetneq N} s_{S}\mathds{1}_{\{i\in S\}}\leq 1, \hspace{0.5em} \forall i\in N_1\\
&& \hspace{-1.5em} \label{eq:LP1C2} \sum_{S\not\subset N_1} s_{S}\mathds{1}_{\{i\in S\}}\leq 0, \hspace{0.5em} \forall i\in N\\ \nonumber
&&  \hspace{-1.5em} ({s}_{S}\geq 0, \forall S\subsetneq N).
\end{eqnarray} Take $\tilde{s}_{S}=\frac{1}{n_1-1}$ $\forall S\subsetneq N_1$, $|S|=n_1-1$, and $\tilde{s}_{S}=0$ for any other $S$. Note that $\{\tilde{s}_{S}\}$ meets~\eqref{eq:LP1C1} and~\eqref{eq:LP1C2} with equality, and thus, it is feasible w.r.t.~\eqref{eq:LP1C1} and~\eqref{eq:LP1C2}. Note, also, that $\sum_{S\subsetneq N_1} |\cap_{i\in N_1\setminus S}\overline{X}^{(1)}_i|\tilde{s}_{S}+\sum_{S\subsetneq N: S\not\subset N_1} |\cap_{i\in N\setminus S}\overline{X}^{(2)}_i| \tilde{s}_{S} = \tilde{r}^{(1)}_N$. By the duality principle, $\{\tilde{r}^{(1)}_i,\tilde{r}^{(2)}_i\}$ is thus optimal w.r.t.~LP~\eqref{eq:SLOO1}. Note that the optimal value is $r_{*}\triangleq\frac{1}{n_1-1}\sum_{i\in N_1}|\overline{X}^{(1)}_i|$. Moreover, the dual of LP~\eqref{eq:SLOO2} is given by 
\begin{eqnarray}\label{eq:dualLP12}
\mathrm{max.} && \nonumber \hspace{-1.5em} \sum_{S\subsetneq N_1} \bigg|\bigcap_{i\in N_1\setminus S}\hspace{-0.25em}\overline{X}^{(1)}_i\bigg|s_{S}+\hspace{-0.25em}\sum_{S\not\subseteq N_1} \bigg|\bigcap_{i\in N\setminus S}\overline{X}^{(2)}_i\bigg| s_{S}+r_{*} s_{*}\\ 
\mathrm{s.t.} && \hspace{-1.5em} \label{eq:LP12C1} \sum_{S\neq N_1} s_{S}\mathds{1}_{\{i\in S\}}+s_{*}\leq 0, \hspace{0.5em} \forall i\in N_1\\
&& \hspace{-1.5em} \label{eq:LP12C2} \sum_{S\not\subseteq N_1} s_{S}\mathds{1}_{\{i\in S\}}\leq 1, \hspace{0.5em} \forall i\in N\\ \nonumber
&&  \hspace{-1.5em} ({s}_{S}\geq 0, \forall S\subsetneq N, S\neq N_1),
\end{eqnarray} where $s_{*}$ is unrestricted in sign. Take $\tilde{s}_{S}=\frac{1}{n_1-1}$ $\forall S\subsetneq N_1$, $|S|=n_1-1$, $\tilde{s}_{N\setminus N_1}=1$, $\tilde{s}_{*}=-1$, and $\tilde{s}_{S}=0$ for any other $S$. Note that $\{\{\tilde{s}_{S}\},\tilde{s}_{*}\}$ meets~\eqref{eq:LP12C1} and~\eqref{eq:LP12C2} with equality. Thus, $\{\{\tilde{s}_{S}\},\tilde{s}_{*}\}$ is feasible w.r.t.~\eqref{eq:LP12C1} and~\eqref{eq:LP12C2}. Note, also, that
$\sum_{S\subsetneq N_1} |\cap_{i\in N_1\setminus S}\overline{X}^{(1)}_i|\tilde{s}_{S}+\sum_{S\not\subseteq N_1} |\cap_{i\in N\setminus S}\overline{X}^{(2)}_i| \tilde{s}_{S}+r_{*}\tilde{s}_{*} = \tilde{r}^{(2)}_N$. By the duality principle, $\{\tilde{r}^{(1)}_i,\tilde{r}^{(2)}_i\}$ is thus optimal w.r.t.~LP~\eqref{eq:SLOO2}, and the optimal value is $|\cap_{i\in N_1} \overline{X}^{(2)}_i|$.
\end{proof}

\bibliographystyle{IEEEtran}
\bibliography{CDERefs}

\end{document}